\documentclass{cagirici}
\usepackage{graphicx,amssymb,amsmath,complexity}

\title{Utilizing hyperplanar substructures to perform efficient range-based WSN localization}

\author{Onur \c{C}a\u{g}{\i}r{\i}c{\i}}
\affil{Faculty of Informatics, Masaryk University Brno, Czech Republic}

\begin{document}
\thispagestyle{empty}
\maketitle

\begin{abstract}
	A wireless sensor network (WSN) consists of multiple wireless sensor nodes that communicate each other to fulfill a particular task.
	In this paper, we emphasize on the networks whose deployments admit lower dimensional substructures, such as collinear groups in 2D, or coplanar groups in 3D.
	When these groups are given as a part of the input, we describe an algorithm to utilize this information to perform a low-cost localization.
	
	In emergency situations such as fire, earthquake etc. inside a building, wireless sensor networks might be very crucial to provide critical information and help the rescue teams to move very quickly by decreasing their burden of exploring the environment.
	Thus, it is very important to develop a system that provides information quickly and without consuming too much energy.
	We observe that in these type of environments, sensor nodes tend to form \emph{hyperplanar groups}.
	A hyperplane is a subspace of one dimension less than its ambient space, and accordingly, a \emph{hyperplanar group} of sensor nodes is a group of nodes that sit on the same hyperplane.
	When we consider a floor of a building, the nodes can be deployed on the corridors to form collinear groups, and when we consider the whole building, the sensor nodes will sit on the floors to form coplanar groups.
	
	Therefore, we study range-based WSN localization problem in 3D environments that induce hyperplanar groupings.
	First, we show that it is an \NP-hard problem to obtain hyperplanar groups, even if we are given the equations of the hyperplanes. Then, we describe an algorithm that assumes each sensor node is aware of its hyperplanar group, and performs a localization when the sensor nodes are deployed in the corridors on each floor of a building. 
	In this case, grouping information allows us to localize a network, which cannot be localized by conventional range-based localization algorithms.
\end{abstract}

\section{Introduction} \label{sec:intro}

Recent  developments in technology pose a need for wireless sensor network technologies to be used broadly \cite{survey}. A \emph{wireless sensor node} is a small device that is equipped with a radio transmitter.
A \emph{wireless sensor network}, in short WSN, is the collection of many wireless sensor nodes.
Each node is able to communicate and measure its distance to another node within the radio's distance.
Independent of the application area, the location information of wireless sensors in a WSN is crucial to improve the quality of service.
The information of positions increases the quality of the applications such as geographical based queries.
When equipped with proper devices, a sensor node can measure its distance to another node in the network using the communication links.
For all types of applications, the information of positions can be appended to the data that is gathered from the sensor nodes, easing the traceability of the sensors, as in WSNs that work in geographical routing.

Range-based wireless sensor network localization problem \cite{laman,rangebased,survey,surveypal,gpslesslowcost} is estimating the positions of the wireless sensor nodes in a network by only using pairwise distances. Estimating the position of a node is called \emph{localizing} the node. In order to localize a node in 3D, we need four pairwise distances from four different localized nodes in general position.

Network localization problem is proved to be a hard problem \cite{laman,survey,complexityOfWSN} but still preferred over Global Positioning System (GPS) \cite{gps} because GPS is usually not applicable for small-scale networks and indoors \cite{gpslesslowcost,zhongthesis,dwrl,reducing,techniques}.

Generally, sensor nodes are assumed to have identical sensing ranges while dealing with range-based localization. Thus, the unit disk graph model is used to model the network. It should be noted that realizing UDG is also an \NP-hard problem, by reduction from {\sc Unit disk graphs recognition} \cite{udg,udgRecognition,unitdiskapprox}. Since range-based localization uses UDG model, one can easily see that it is also \NP-hard \cite{complexityOfWSN}.

In \cite{cagiricithesis}, it was observed, and  experimentally verified that knowing which nodes are in the same group increases the accuracy, and decreases the computational time, when the number of groups is much smaller than the number of sensors. This is because conventional distance-based localization methods, such as trilateration \cite{trilat} and quadrilateration \cite{cagiricithesis}, use distances from points (in our case sensors) in general position. However, in practice, objects lie on planar surfaces to form coplanar groups. Which means a large amount of computational power is spent on finding correct nodes to localize another node. 

In an effort to perform a more efficient localization, we can exploit the information on the hyperplanar substructures, if provided. That is, a function $h : S \to \{1, 2, \dots, k\}$, called \emph{hyperplanar grouping function}. Two sensor nodes $s_i$ and $s_j$ are said to be on the same hyperplanar group if, and only if, $h(i) = h(j)$. The task of finding hyperplanar grouping function in 2D from given UDG, and hyperplanes is called \emph{hyperplanar substructures recognition}. 
In case where hyperplanar substructures are realized, \emph{i.e.} the mapping $h$ is given, we are able to perform a lower dimensional localization for each hyperplanar group separately, and then localize the groups globally with respect to each other.

In this paper, we show that discovering the mapping function $h$ is an \NP-hard problem, even if the hyperplanes are given as a part of the input by their equations. Then, we describe an algorithm to perform an efficient localization in case where $h$ is given as input.
Afterwards, we describe an algorithm to utilize the grouping function, when available, and show that our algorithm is able to localize a network which cannot be localized by using only pairwise distances.

\subsection{Motivation}
The physical world we live in is a 3D environment. Therefore, many applications require wireless noeds to be localized in 3D.
Most of the scenarios in real-world WSN applications that need localization have usually deployments where the sensor nodes sit on planar regions to form sets of planar clusters.
Examples of these surfaces include structures seen in both indoor (\textit{e.g.} floors, doors, walls, tables etc.) and outdoor (\textit{e.g.} mountains, valleys, hills etc.) environments.
It is hence observed that a method for exploiting the information of structural information in 3D is very much needed in order to improve the quality of localization achieved.

In this paper, we particularly study the scenario where there is an emergency inside a building, such as fire, earthquake, hostage situation etc.
When the members of a search and rescue team have limited time, our work aims to utilize wireless sensor networks to release the burden of searching, and allowing team to go for rescuing.

Suppose that the wireless nodes are deployed onto the corridors of a building. 
Each corridor can be interpreted as a collinear group (i.e. hyperplanar group in 2D).
Moreover, since the building probably has multiple floors, each floor is a coplanar group (i.e. hyperplanar group in 3D).

In Figure~\ref{fig:deployment}, we give an example of a network deployed onto the corridors of a three-floor building.
The dots denote the wireless nodes inside the building, and colored rectangles denote three different floors.
A pair of nodes denoted by the same color sit on the same floor.

\begin{figure}
	\centering
	\begin{tikzpicture}
	\begin{axis}
	[ axis on top, view={-30}{-15}	]
	\addplot3[draw=none, fill=blue!50]
	coordinates{
		(0,0,0)
		(0,1,0)
		(1,1,0)
		(1,0,0)
	};
	\addplot3[draw=none, fill=yellow!50]
	coordinates{
		(0,0,0.8)
		(0,1,0.8)
		(1,1,0.8)
		(1,0,0.8)
	};
	\addplot3[draw=none, fill=red!20]
	coordinates{
		(0,0,1.6)
		(0,1,1.6)
		(1,1,1.6)
		(1,0,1.6)
	};
	\addplot3[scatter, only marks] file {./sensors.dat};
	\end{axis}
	\end{tikzpicture}
	\caption{Sensor nodes deployed onto the corridors of a building}
	\label{fig:deployment}
\end{figure}
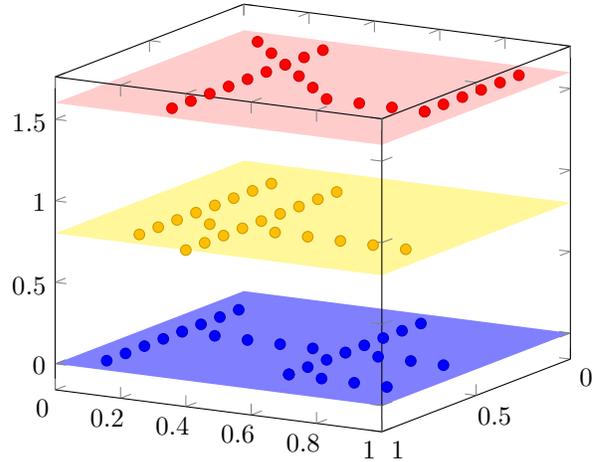

\subsection{Basic terminology}

In this section, we give the basic terminology that we use in the described algorithm.

\emph{Localizing} of a sensor network is assigning relative positions to each sensor node in the ambient space (in our case, 3D).
When we assign positions to each node, we obtain a \emph{point formation}. The point formation of a network is denoted by $\mathbb{F}$. 
It is a table that contains $2$ columns, and $n$ rows, where $n$ is the number of nodes in the network.
The first column contains the IDs of the nodes, and the second column contains each node's position in 3D Euclidean space.

A wireless sensor network is said to be \emph{localizable} in $\mathbb{R}^d$ when the corresponding distance graph $G = (V, E)$ is \emph{globally rigid}.
In other words, every embedding of the vertices in $V$ with respect to the pairwise Euclidean distances determined by the edge set $E$ is isometric in $\mathbb{R}^d$.
For more details on the localizability, we refer the reader to \cite{theory,understandinglocalizability,localizabilityanddistributed}.

Every sensor node in the given network is initially \emph{unlocalized}, i.e. it is not aware of its position.
When we fix the position of a node, we mark it as \emph{localized}.
To fix the position, we need to reduce its \emph{degrees of freedom} to zero by obtaining $d+1$ Euclidean distances from its localized neighbors in general position.
Intuitively, one expects a localizable network to be dense enough.
Our algorithm, however, is able to perform localization in very sparse graphs if the grouping function is provided.

Observe that each hyperplanar group is an induced subgraph of the network graph $G = (V, E)$, and each pair $(u,v)$ of vertices in that induced subgraph satisfies $h(u) = h(v)$.
The set of all vertices in a particular hyperplanar group can be retrieved by taking the inverse of $h$. 
In other words, $h^{-1}(i)$ gives us the set of all vertices in the hyperplanar group $i$.

We refer to the hyperplanar groups in 2D and 3D as \emph{collinear groups}, and \emph{coplanar groups} respectively.
An edge is called \emph{interlinear} (respectively, \emph{interplanar}) edge, if it corresponds to a connection between two nodes from different collinear (respectively, \emph{coplanar}) groups.
Otherwise, it is a collinear (respectively, \emph{coplanar}) edge.

\subsection{Our contributions}	

Our main contributions are listed as follows.

\begin{itemize}
	\item We exploit the fact that the sensor nodes tend to form hyperplanar groups in real-world scenarios.
	\item We prove that extracting the grouping information is a hard task.
	\item We give an efficient and sustainable localization algorithm to utilize provided grouping function. Our algorithm does not require a dense network graph.
\end{itemize}

\section{Hyperplanar substructures recognition} \label{sec:reduction}
In this section, we formally define the problem of recognizing hyperplanar substructures, and then prove that the problem is \NP-hard. 

\subsection{Problem definition}
Consider a set $\mathcal{H} = \{\ell_1, \ell_2, \dots, \ell_k\}$ of $k$ hyperlanes in $\mathbb{R}^d$, given by their equations $\ell_i : a_{i1}x_1 + a_{i2}x_2 + \dots + a_{in}x_d = b$, and a unit disk graph $G = (V,E)$, where $V = \{1,2,\dots,n\}$.
The task is to find a mapping $h : V \rightarrow \{1,\dots,k\}$ such that there is a realization of $G$ in which $v \in \ell_{f(v)}$ for each $v \in V$.

\subsection{NP-hardness reduction for 2D environment} \label{sec:hard2d}
We first show that Hyperplanar substructures recognition is \NP-hard when the sensor nodes are in 2D, and then move onto higher dimensions.

\begin{lem} \label{lem:hard2d}
	Hyperplanar substructures recognition is \NP-hard in 2D.
\end{lem}

\begin{proof}
	We give a reduction of this problem from 2-coloring of 3-uniform hypergraphs. 
	The polynomial-time reduction from a given graph to is as follows.
	Let $H = (X,F)$ be an instance of a 3-uniform hypergraph. $H$ is said to be 2-colorable if there is an assignment of two colors with no monochromatic edge.
	In the following part, when we write ``uniformly distributed,'' we mean the Euclidean distance between two sensor nodes is less than one unit.
	We construct the gadget to perform the reduction as follows. 
	\begin{enumerate}
		\item We have $|X|$ vertical, and $2|F|+2$ horizontal straight lines. Each of them are parallel to either $x$-axis or $y$-axis.
		
		\item All the vertical lines are colored black, and spaced equally. These lines are called \emph{vertex lines}, that represent the vertex set of $H$. The Euclidean distance between each pair of vertical line is greater than two units, and less than three units.
		
		\item One of the horizontal lines, called the \emph{main line}, is given by the equation $y = 0$. This line is colored black. There are sufficient number of uniformly distributed wireless sensor nodes on the main line, including the intersection points of the main line and the vertex lines.
		
		\item In order to fix the configuration, there exists a \emph{support line}, which is placed horizontally, and sufficiently close the main line. This line contains only two vertices that forms a $K_4$ with two vertices from the main line.
		
		\item Beginning from one unit above and below the main line, there are uniformly distributed sensor nodes on the vertex lines.
		
		\item Remaining $2|F|$ horizontal lines are distributed equally above and below the main line. The ones above the main line are colored red. The rest of the horizontal lines are colored blue. All the colored lines are sufficiently far from the main line. Each $r_i$ and $b_i$ represent two mirrored copies of an edge $i$ of $H$, and we refer to them as \emph{edge lines}. We, later on, will use the edge lines to determine the color of an edge in $H$. 
		
		\item The Euclidean distance between each consecutive pairs of horizontal lines with the same color are equal, and greater than one unit.
		
		\item There exist wireless nodes on the edge lines, each neighboring exactly two sensor nodes from a vertex line, forming a triangle. These triangles are referred to as \emph{flags}, and the vertices are referred to as \emph{flag vertices} in the following part of the proof.
		
	\end{enumerate}
	
	In Figure~\ref{fig:nphardness}, we draw the described configuration. The support line is indicated with dashes. The main line is just below the support line. There are six vertex lines, and six edge lines. That means in the hypergraph $H = (X, F)$, $|X| = 6$ and $|F| = 3$. Since every wireless node on the main line and the vertex lines are distributed uniformly, and the Euclidean distance between parallel lines is always greater than one unit, it is easy to fix the groups of the nodes on the main line and the vertex lines. Moreover, notice the flags on the leftmost and the rightmost vertex lines. These flags help us to fix the flag vertices those contain them.
	
	All fixed nodes are indicated with white circles. The remaining nodes, flag nodes, are indicated with black circles. Now, our task is to find which flag nodes form a collinear group on an edge line.
	
	The edge lines are two symmetrical copies of the edges in $F$ with different colors with respect to the main line. We denote the two copies of an edge $f \in F$ as $b_f$ and $r_f$, for blue and red colors respectively.
	
	If a vertex line that represents $x \in X$ has its missing flag on an edge line $b_f$ or $r_f$, then this indicates that the vertex $x \in X$ is covered by the edge $f \in F$ in the hypergraph $H$.
	In that case, the color of $x$ is determined by flipping the configuration (either upside down or left to right) on the vertex line that corresponds to $x$.
	As demonstrated in Figure~\ref{fig:colorfix}, the colors of some vertices in $H$ might be determined by the vertices in $V$. In this particular example, $u,v \in V$ are adjacent vertices, and $w$ is neighboring none of them. Therefore, neither $u$ nor $v$ can flip horizontally, and they must flip vertically together.
	This means that the color of the vertices $x,y \in X$ covered by the edge $f \in F$ in $H$ is determined by the vertices $u, v, w \in V$.
	
	On each vertex line, there are five flags. If an edge has a ``missing'' flag on one side then that means that edge covers the vertex, and is colored according to the color of the edge line which is missing the flag.
	Remember that the Euclidean distance between consecutive vertex lines is less than $3$ units. Hence, if two flag nodes are facing to each other, they must be neighbors in our configuration.
	Given this information, there can be only one mapping for hyperplanar grouping function for the input UDG $G = (V,E)$, and only one corresponding coloring of its corresponding hypergprah $H = (X,F)$.
	
	Therefore, an algorithm to solve the problem in polynomial time can also be used to solve 2-coloring of 3-uniform hypergraphs in polynomial time. 	
\end{proof}

\begin{figure}[htbp]
	\centering
	\begin{tikzpicture}[scale=0.25]
	
	\node at (-14, 15) {\footnotesize fixing lines};
	\draw[dashed,->] (-14,14.5)--(-10,12);
	\draw[dashed,->] (-14,14.5)--(10,12);
	\node at (-15,2) {\footnotesize main line};
	\node at (-15,1.2) {\footnotesize and support line};
	\tikzstyle{every node}=[draw, fill=white, shape=circle, minimum size=3pt,inner sep=0pt];
	\draw[thick] (-16,0) -- (13,0);
	\draw[dashed] (-16,0.5) -- (13,0.5);
	\draw (-15,0)--(-16.5,0.5)--(-15.5,0.5) -- (-14,0) -- (-16.5,0.5);
	\draw (-15,0)--(-15.5,0.5);
	\node at (-16.5,0.5) {};
	\node at (-15.5,0.5) {};

	\foreach \i in {8.5, 6.5, 4.5}
	{
		\draw[color=red] (-12,\i) -- (12,\i);
		\draw (-11, \i) -- (-10, \i + .5) -- (-9, \i) -- (-10, \i - .5) -- (-11, \i);
		\draw (11, \i) -- (10, \i + .5) -- (9, \i) -- (10, \i - .5) -- (11, \i);
		\node at (-11, \i) {};
		\node at (-9, \i) {};
		\node at (11, \i) {};
		\node at (9, \i) {};
	}

	\foreach \i in {-4.5, -6.5, -8.5}
	{
		\draw[color=blue] (-12,\i) -- (12,\i);
		\draw (-11, \i) -- (-10, \i + .5) -- (-9, \i) -- (-10, \i - .5) -- (-11, \i);
		\draw (11, \i) -- (10, \i + .5) -- (9, \i) -- (10, \i - .5) -- (11, \i);
		\node at (11, \i) {};
		\node at (9, \i) {};
		\node at (-11, \i) {};
		\node at (-9, \i) {};
	}
	
	\node[draw=none] at (-13,8.5) {$r_3$};
	\node[draw=none] at (-13,6.5) {$r_2$};
	\node[draw=none] at (-13,4.5) {$r_1$};
	\node[draw=none] at (-13,-8.5) {$b_3$};
	\node[draw=none] at (-13,-6.5) {$b_2$};
	\node[draw=none] at (-13,-4.5) {$b_1$};

	\foreach \i in {-10,-6,-2,2,6,10} 
	{
		\draw (\i,12) -- (\i,-12);
		\foreach \j in {2,...,10}
		{
			\node (\i\j) at (\i, \j) {};
		}
		\foreach \j in {-2,...,-10}
		{
			\node (\i\j) at (\i, \j) {};
		}
	}
	
	\foreach \i in {-15,...,12}
	\node at (\i, 0) {};

	\node[draw=none] at (-6,-13) {$a$};
	\node[draw=none] at (-2,-13) {$b$};
	\node[draw=none] at (2,-13) {$c$};
	\node[draw=none] at (6,-13) {$d$};

	\tikzstyle{every node}=[draw, fill=black, shape=circle, minimum size=2.5pt,inner sep=0pt];
	\node (r1) at (-5, 8.5) {};
	\draw (r1)--(-69);
	\draw (r1)--(-68);
	
	\node (r2) at (-3, 8.5) {};
	\draw (r2)--(-29);
	\draw (r2)--(-28);
	
	\node (r3) at (-3, 6.5) {};
	\draw (r3)--(-27);
	\draw (r3)--(-26);
	
	\node (r4) at (-3, 4.5) {};
	\draw (r4)--(-25);
	\draw (r4)--(-24);
	
	\node (r5) at (1, 6.5) {};
	\draw (r5)--(27);
	\draw (r5)--(26);
	
	\node (r6) at (5, 4.5) {};
	\draw (r6)--(65);
	\draw (r6)--(64);

	\node (b1) at (-5, -8.5) {};
	\draw (b1)--(-6-9);
	\draw (b1)--(-6-8);
	
	\node (b2) at (-5, -6.5) {};
	\draw (b2)--(-6-7);
	\draw (b2)--(-6-6);
	
	\node (b3) at (-5, -4.5) {};
	\draw (b3)--(-6-5);
	\draw (b3)--(-6-4);
	
	\node (b4) at (3, -8.5) {};
	\draw (b4)--(2-9);
	\draw (b4)--(2-8);
	
	\node (b5) at (3, -6.5) {};
	\draw (b5)--(2-7);
	\draw (b5)--(2-6);
	
	\node (b6) at (3, -4.5) {};
	\draw (b6)--(2-5);
	\draw (b6)--(2-4);
	
	\node (b7) at (5, -8.5) {};
	\draw (b7)--(6-9);
	\draw (b7)--(6-8);
	
	\node (b8) at (5, -6.5) {};
	\draw (b8)--(6-7);
	\draw (b8)--(6-6);
	
	\node (b9) at (5, -4.5) {};
	\draw (b9)--(6-5);
	\draw (b9)--(6-4);

	\end{tikzpicture}
	\caption{Gadget to prove NP-hardness}
	\label{fig:nphardness}
\end{figure}
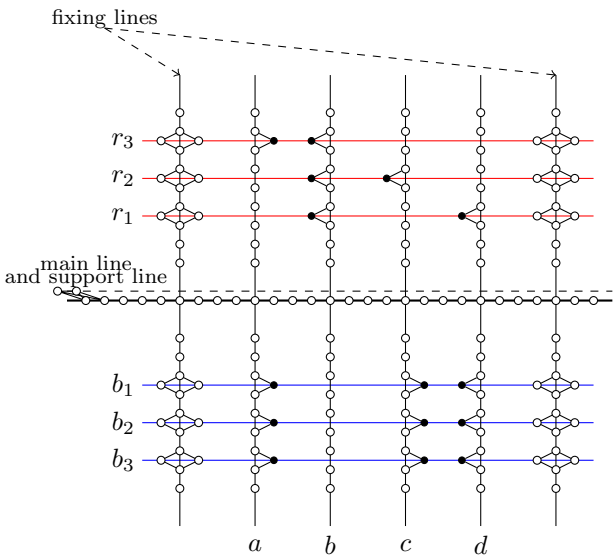

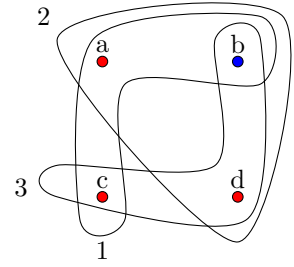
\begin{figure}[htbp]
	\centering
	\begin{tikzpicture}[scale=0.6]

	\tikzstyle{every node}=[draw, shape=circle, minimum size=4pt,inner sep=0pt];
	\node[label=a, fill=red] at (0,3) {};
	\node[label=b, fill=blue] at (3,3) {};
	\node[label=c, fill=red] at (0,0) {};
	\node[label=d, fill=red] at (3,0) {};
	\node[draw=none] at (0,-1.2) {1};
	\node[draw=none] at (-1.3,4) {2};
	\node[draw=none] at (-1.8,0.2) {3};
	\draw plot [smooth cycle] coordinates {(-0.2,3.5) (3.5,4) (3.5, 2.5) (0.5,2.5) (0.5, -0.5) (-0.5, -0.5)};
	\draw plot [smooth cycle] coordinates {(2.5,3.5)  (2.5,0.7) (-1,0.7) (-1,0) (3.2, -0.5) (3.5, 3.5)};
	\draw plot [smooth cycle] coordinates {(-1,3.5) (3,-1) (4,4)};

	\end{tikzpicture}
	\caption{Corresponding hypergraph of the configuration given in Figure~\ref{fig:nphardness}}
	\label{fig:hypergraph}
\end{figure}

\begin{figure}[htbp]
	\centering
	\begin{tikzpicture}[scale=0.2]
	
	\tikzstyle{every node}=[draw, fill=white, shape=circle, minimum size=4pt,inner sep=2pt];
	\draw (6,-12)--(6,12);
	\draw (0,-12)--(0,12);
	\draw (-6,-12)--(-6,12);
	
	\node[draw=none] at (0,-13) {$x$};		
	\node[draw=none] at (6,-13) {$y$};

	\draw[dotted] (10, -4)--(10, -7);
	\draw[dotted] (10, 4)--(10, 7);
	\draw[dotted] (-10, -4)--(-10, -7);
	\draw[dotted] (-10, 4)--(-10, 7);
	
	\node[draw=none] at (-14,4) {$r_f$};
	\node[draw=none] at (-14,-4) {$b_f$};
	
	\draw[color=red] (-12,4)--(12,4);
	
	\draw[color=red] (-12,7)--(12,7);
	\draw[color=red] (-12,8)--(12,8);
	\draw[color=red] (-12,9)--(12,9);
	\draw[color=red] (-12,10)--(12,10);
	
	\draw[color=blue] (-12,-4)--(12,-4);
	
	\draw[color=blue] (-12,-7)--(12,-7);
	\draw[color=blue] (-12,-8)--(12,-8);
	\draw[color=blue] (-12,-9)--(12,-9);
	\draw[color=blue] (-12,-10)--(12,-10);
	
	\node (w) at (-4.5,-4) {$w$};
	\node (u) at (1.5,-4) {$u$};
	\node (v) at (4.5,-4) {$v$};
	
	\node (w1) at (-6, -2) {};
	\node (w2) at (-6, -6) {};
	\node (u1) at (0, -2) {};
	\node (u2) at (0, -6) {};
	\node (v1) at (6, -2) {};
	\node (v2) at (6, -6) {};
	
	\draw (u1)--(u)--(u2);
	\draw (v1)--(v)--(v2);
	\draw (w1)--(w)--(w2);
	\end{tikzpicture}
	\caption{The color of the edge that covers vertex $u$ is forced by vertices $v$ and $w$.}
	\label{fig:colorfix}
\end{figure}
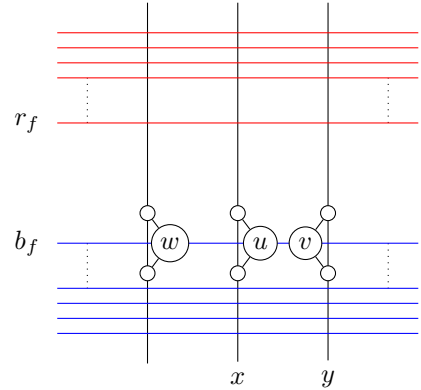

\subsection{NP-hardness reduction for every dimension} \label{sec:hard3d}
In Section~\ref{sec:hard2d}, we have proved that the mentioned problem is \NP-hard in 2D by a reduction from 2-coloring of 3-uniform hypergraphs. Now, let us give the proof for the case where the wireless sensors are deployed in 3D.
\begin{lem} \label{lem:hard3d}
	Hyperplanar substructures recognition is \NP-hard in 3D.
\end{lem}
\begin{proof}
	This reduction is hyperplanar substructures realization in 2D. 
	Consider the gadget given in Figure~\ref{fig:nphardness} from Section~\ref{sec:hard2d}.
	Suppose that each line is a plane seen from $xy$-plane in 3D. That is, each wireless node in Figure~\ref{fig:nphardness} has $z$-coordinate zero.
	In order to obtain a deployment in 3D, we just make a copy for each wireless node with the same $x$ and $y$-coordinates, and let their $z$-coordinates be equal to $1$.
	
	From this point on, it is trivial to see that each wireless node is connected to its copy in 3D, and no other from the copies.
	Thus, the triangle flags become triangle prisms.
	And we can easily apply the proof given in Section~\ref{fig:nphardness}. 
\end{proof}

Now that we have proved that proof given in Lemma~\ref{lem:hard2d} can be extended to 3D by Lemma~\ref{lem:hard3d}, we can use the same method and extend the claim to higher dimensions.
\begin{thm} \label{thm:hardAllDimensions}
	Hyperplanar substructures recognition is \NP-hard in every dimension.
\end{thm}
\begin{proof}
	We already know that the problem is \NP-hard in 2D by Lemma~\ref{lem:hard2d}, and the hardness reduction can be extended to 3D by Lemma~\ref{lem:hard3d}.
	In this proof, we show that the very same reduction can be used to prove the claim for 4D, and recursively, to higher dimensions.
	
	In 2D, every node has 2 components, namely $x$ and $y$. If we include one one component, $z$, and assign $1$ to all third components of the sensors, Lemma~\ref{lem:hard2d} remains true.
	
	Every component, in the context of this paper, adds an extra dimension to the nodes.
	Thus, for higher dimensions, we are able to set the new component's value to $1$. The new copies of the sensor nodes will not be neighbors of the older ones. Also, the flags will be preserved.
	Therefore, the problem remains \NP-hard in higher dimensions.
\end{proof}

\section{Localization process}
In this section, we discuss the details of utilizing the grouping function.
First, we describe an algorithm to localize two collinear groups with respect to each other.
Then, we use the same idea to localize two coplanar groups with respect to each other.
Finally, we show that our algorithm is able to perform localization easily where the quadrilateration algorithm fails.

For the sake of completeness, we give the quadrilateration algorithm in Algorithm~\ref{alg:quadrialteration}.
We would like to point out that algorithms given in this section do not use any GPS, and therefore they can only produce a point formation isometric (i.e. up to sign ambiguity) to the original network in the best case.

\begin{algorithm}[htbp]
	\begin{algorithmic}[1]
		\Input{Wireless sensor network graph $G = (V, E)$}
		\Output{3D point formation $\mathbb{F}$ of $G$}
		\State $\mathbb{F} \gets [\ ]$ \label{quadInitPF}
		\State Locate a non-coplanar $K_4$ $\{i,j,k,l\} \subset V$ \label{seednodes}
		\State Assign positions to $i,j,k,l$ \label{assignpositions}
		\State $Q_{\textit{localized}} \gets [i,j,k,l]$ \label{quadInitQ}
		\While{$Q_{\textit{localized}}$ is not empty} \label{quadQnotEmpty}
		\State $u \gets \textit{dequeue}(Q_{\textit{localized}})$\label{quadDequeue}
		\State add $(u, u.\textit{pos})$ into $\mathbb{F}$ \label{quadAddIntoF}
		\For{all neighbors $v$ of $u$} \label{quadIterateOnNeighbors}
		\If{$(\texttt{++}v.\textit{Count}) \geq 4$} \label{quadCountge4}
		\If{$v$ can be localized in 3D} \label{canBeLocalized}
		\State $\textit{enqueue}(v,Q_{\textit{localized}})$ \label{quadAddIntoQ}
		\EndIf
		\EndIf
		\EndFor
		\EndWhile
		\State \Return $\mathbb{F}$ \label{quadReturnBest}
	\end{algorithmic}
	\caption{Quadrilateration algorithm}
	\label{alg:quadrialteration}
\end{algorithm}

Basically, the algorithm takes the network graph $G = (V, E)$ as the input, and assigns positions to each vertex in $v \in V$, based on the pairwise distances given by $E$.
We begin by initializing the point formation in Line~\ref{quadInitPF}.
In Line~\ref{seednodes}, we find a non-coplanar $K_4$ in $G$ to begin the localization.
The vertices of this complete subgraph are called \emph{the seeds}.
Since every pairwise distance is available, we can assign positions to these vertices in Line~\ref{assignpositions}.
In Line~\ref{quadInitQ}, we initialize a queue, called $Q_{\textit{localized}}$, that contains the localized vertices. Initially, it contains only the seeds.
In Line~\ref{quadQnotEmpty}, we start a loop, and continue until $Q_{\textit{localized}}$ is empty.
In the loop, we first dequeue the first node in Line~\ref{quadDequeue}, and then store its position into the point formation $\mathbb{F}$ in Line~\ref{quadAddIntoF}.
Now that we have a localized vertex $u$, we iterate all its neighbors and broadcast the position of $u$.
In Line~\ref{quadCountge4}, we use a variable $v.\textit{Count}$, which indicates the number of localized neighbors of $v$.
Thus, if some neighbor $v$ of $u$ has at least $4$ localized neighbors, we check if it can be localized with respect to these neighbors in Line~\ref{canBeLocalized}.
Remember that if the localized neighbors are coplanar, then we cannot fix the position of $v$.
If $v$ can be localized, then we add it to $Q_{\textit{localized}}$, and continue the algorithm until the queue is empty.
Finally, we return the point formation in Line~\ref{quadReturnBest}.

\textbf{Complexity analysis.} 
In a network graph with $n$ vertices and $m$ edges, Algorithm~\ref{alg:quadrialteration} runs in $\OO(n^4m)$ time in the worst case.

Notice that even though the network graph admits a proper ordering to run the algorithm, the first $K_4$ we find in Line~\ref{seednodes} might not be the correct vertices to start \cite{theory}.
As four distance measurements from four vertices in general position are sufficient to localize a node, $\OO(m)$ time is spent in each phase by using the queue data structure keeping the track of localized nodes.
This can be accomplished by keeping the track of the number of localized neighbors of every node and inserting a node when the count hits the score of three.
An effective three is the smallest count $\geq 4$ with three non-coplanar neighbors.
When we reach the count $4$ for the first time, we check coplanarity.
From then on, every new increment checks only the new one with any two previous neighbors.
Therefore, it terminates after at most $\OO(n^4m)$ steps.

As seen above, the worst-case analysis gives us quartic running time.
But notice that when the sensor nodes are deployed to form coplanar groups in 3D, and the communication between groups are held by a limited number of nodes (i.e. the number of interplanar edges are low), most of the neighbors of an unlocalized node will be coplanar.
On the other hand, in the algorithm that we describe to utilize grouping function, this type of deployment is an advantage.
Thus, it is safe to say that these two algorithms complete each other, in terms of their functionalities with respect to sparsity.

While describing our algorithms, we use the following data structure for a vertex of the network graph.
A sensor node $s_u$, represented by a vertex $u$ in the network graph has three positions for each dimension. 
Namely, $\textit{pos}_1$, $\textit{pos}_2$, and $\textit{pos}_3$.
$u.\textit{pos}_1$ is a number that indicates the relative position of $u$ with respect to other vertices in the same collinear group.
Similarly $u.\textit{pos}_2$ is the position of $u$ with respect to the other vertices in the same coplanar group.
Finally, $u.\textit{pos}_3$ is the actual position of $u$ in $\mathbb{R}^3$.

\subsection{Localizing the nodes of a collinear group}
In this section, we describe a linear-time algorithm to localize the nodes in a given collinear group ambient in 2D.
The induced graph of a collinear group is called a \textit{unit interval graph} since all the vertices are considered to be in 1D.
The algorithm takes an induced subgraph of the network graph as the input, and outputs the embedding of that subgraph onto a straight line.
For the sake of simplicity, we embed the vertices onto $x$-axis.
While describing the algorithm, we denote the subgraph as $G' = (V', E')$.
We also assume that $G'$ is a connected graph.

If there are no cycles in a collinear group, then we can assign position to each vertex in linear time.
The algorithm to achieve this is pretty straightforward.

\begin{algorithm}[htbp]
	\begin{algorithmic}[1]
		\Input{Simple path $P = (v_1, v_2, \dots, v_k)$ of size $k$}
		\Output{1D point formation $\mathbb{F}$ of $P$}
		\State $v_1.\textit{pos}_1 \gets 0$
		\For{$i = 1 \to k$}
		\State $u \gets v_i$
		\State $v.\textit{pos} \gets w(uv)$
		\State Mark $u$ as localized
		\State Add $u.\textit{pos}_1$ into $\mathbb{F}$
		\EndFor
		\State \Return $\mathbb{F}$			
	\end{algorithmic}
	\caption{Localization algorithm for a collinear group with no cycles}
	\label{alg:intervalgraph}
\end{algorithm}

Algorithm~\ref{alg:intervalgraph} runs in $\OO(n)$ time where $n$ is the number of vertices in a collinear group.
We simply process all the vertices once and assign positions until there are no vertices left.
This algorithm is a strong indication that the information on grouping function reduces the number of computations dramatically.

In case where the collinear group contains cycles, we can use the algorithm proposed by Brandst{\"a}dt et al. \cite{linearhamiltonian} to find the Hamiltonian path in (claw,net)-free graphs. A (claw,net)-free graph does not contain both of the graphs given in Figure~\ref{fig:claw} as an induced subgraph.

\begin{figure}[htbp]
	\centering
	\begin{tikzpicture}[scale=1.1]
	\tikzstyle{every node}=[draw, fill=white, shape=circle, radius=3pt,inner sep=1.5pt];
	\node (0) at (0,0) {$a$};
	\node (1) at (0,1) {$b$};
	\node (2) at (1,-0.5) {$d$};
	\node (3) at (-1,-0.5) {$c$};
	\draw (0)--(1);
	\draw (0)--(2);
	\draw (0)--(3);
	
	\node (4) at (4,0.2) {$a$};
	\node (5) at (3.7,-0.3) {$b$};
	\node (6) at (4.3,-0.3) {$c$};
	\draw (4)--(5)--(6)--(4);
	
	\node (7) at (4, 1) {$x$};
	\node (8) at (3, -0.5) {$y$};
	\node (9) at (5, -0.5) {$z$};
	\draw (4)--(7);
	\draw (5)--(8);
	\draw (6)--(9);
	\end{tikzpicture}
	\caption{A claw (on the left) and a net (on the right)}
	\label{fig:claw}
\end{figure}
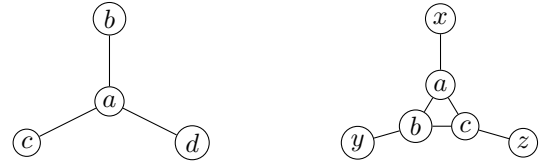

Let us first show that interval graphs are (claw,net)-free graphs.

\begin{lem}
	Every unit interval graph is claw-free.
\end{lem}

\begin{proof}
	Suppose that $K(a;b,c,d)$ is an induced claw, as shown in Figure~\ref{fig:claw}.
	Since $b$ and $c$ are not neighbors, that means the Euclidean distance between them is greater than $1$ unit.
	And also, both of them are neighbors with $a$, which means the Euclidean distance between them is less than $2$ units.
	If we embed $b$, $a$, $c$ onto a straight line respectively, then $a$ must be between $b$ and $c$.
	
	Wherever we place $d$, within one unit distance from $a$, it must be a neighbor of either $b$ or $c$.
	If $d$ is neighbors with neither, then the Euclidean distance between $d$ and $a$ must be greater than one unit, contradicting the unit interval property.
\end{proof}

\begin{lem}
	Every unit interval graph is net-free.
\end{lem}
\begin{proof}
	Suppose that $N(a, b, c; x, y, z)$ is an induced net, as shown in Figure~\ref{fig:claw}.
	It is clear that $a$, $b$ and $c$ must be inside an interval of at most $1$ unit.
	Note that $a,x$ and $c,z$ are pairwise disjoint neighbors, and moreover neither $x$ nor $z$ is a neighbor of $b$.
	This means that $b$ must be between $a$ and $c$.
	
	Wherever we place $y$, within one unit distance from $b$, it must be a neighbor of either $a$ or $c$ (also, at least one of $x$, $z$).
	If $y$ is a neighbor with neither, then the Euclidean distance between $b$ and $y$ must be greater than one unit, contradicting the unit interval property.
\end{proof}

\begin{cor}
	Every unit interval graph is (claw,net)-free.
\end{cor}

When a Hamiltonian path is found, Algorithm~\ref{alg:intervalgraph} can use the induced path as input, and assign positions to the vertices in linear time.

\begin{thm}[\cite{linearhamiltonian}] \label{thm:hamiltonian}
	Every (claw,net)-free graph $G$ has a Hamiltonian path and such a path can be found in $\OO(n + m)$ time.
\end{thm}

By Theorem~\ref{thm:hamiltonian}, we can use the algorithm given in \cite{linearhamiltonian}, and find a Hamiltonian path in a collinear group.
Then, we can simply give the found path as the input of Algorithm~\ref{alg:intervalgraph}, and obtain the relative positions of vertices in a collinear group in $\OO(m + n)$ running time.

\subsection{Localizing collinear groups with respect to each other}
We have proposed a method to localize the vertices of a collinear group in 1D.
Now, we describe an algorithm to find the position of two hyperplanar groups in their ambient space.
This algorithm is described in general for $\mathbb{R}^d$, and therefore is applicable to both 2D and 3D.

It fixes the position of $d$ vertices in a hyperplanar group, and then uses these positions to determine the equation of the hyperplane on which the vertices of that group lies.
Finally, the positions of every other vertex in the same group can be found in the same ambient space by applying a transformation.

The hyperplanar group of a vertex $x$ is denoted by $h(x)$, and we can retrieve the set of vertices $U$ where for each pair $u,v \in U$, $h(u) = h(v)$ by taking the inverse of $h(\cdot)$ (denoted by $h^{-1}(\cdot)$).
Note that we can use $h$ for both localization of collinear groups or coplanar groups without loss of generality.
However, we operate only in one dimenstion i.e. we do not localize a collinear group with respect to a coplanar group, and vice versa.
Similarly, when we write $u.\textit{pos}_d$, it is the relative position of $u$ in $d$ dimensional space.

The equation of the hyperplane that every vertex in the group $i$ lies is denoted by $\textit{eqn}(i)$.
We say that a hyperplanar group $i$ is adjacent to $j$, if there exists at least one interlinear (or interplanar) edge between a vertex in $i$, and a vertex in $j$.

\begin{algorithm}[htbp]
	\begin{algorithmic}[1]
		\Input{Network graph $G = (V, E)$, hyperplanar grouping function $h$}
		\Output{Point formation $\mathbb{F}$ of $G$ in $\mathbb{R}^d$}
		\State Find $v.\textit{pos}_1$ for each $v \in V$
		\State Pick a group $i$
		\State Mark the vertices in $i$ as localized
		\For{each group $j$ adjacent to $i$}
		\For{each $u \in h^{-1}(j)$}
		\If{$u$ has $d$ neighbors in $i$}
		\State Determine $u.\textit{pos}_d$
		\State Mark $u$ as support vertex
		\If{There are $d$ support vertices}
		\State Determine $\textit{eqn}(j)$
		\State Find the transformation matrix
		\State Apply the transformation to $h^{-1}(j)$
		\State Add the positions into $\mathbb{F}$
		\State Notify the neighbors of localized vertices
		\State \textbf{go to} Line~\ref{isLocalizable}
		\EndIf
		\EndIf
		\EndFor
		\EndFor
		\If{No collinear groups are localized} \label{isLocalizable}
		\State \Return{} {\sc not\_localizable}
		\EndIf
		
		\For{each non-localized group $j$}
		\For{each $u \in h^{-1}(j)$}
		\If{$u$ was notified $d+1$ times}
		\State Determine $u.\textit{pos}_d$
		\State Mark $u$ as support vertex
		\If{There are $d$ support vertices}
		\State Determine $\textit{eqn}(j)$
		\State Find the transformation matrix
		\State Apply the transformation to $h^{-1}(j)$
		\State Add the positions into $\mathbb{F}$
		\State \textbf{continue} the iteration
		\EndIf
		\EndIf
		\EndFor
		\EndFor
		\State \Return $\mathbb{F}$
	\end{algorithmic}
	\caption{Localizing hyperplanar groups with respect to each other}
	\label{alg:hyperplanarlocalization}
\end{algorithm}

\textbf{Complexity analysis.}
In a network graph with $n$ vertices, $m$ edges, $k$ coplanar groups, and $r$ collinear groups, Algorithm~\ref{alg:hyperplanarlocalization} runs in $\OO(k^2m) + \OO(r^2m) + O(n+m)$ time in the worst case.

We spend $\OO(n+m)$ time to find the 1D positions of the vertices using Algorithm~\ref{alg:intervalgraph} \cite{linearhamiltonian}.
Note that we iterate on each collinear group only once, therefore the total running time is not affected by the number of the collinear groups.
For each collinear group, we iterate all the other collinear groups to find proper edges to localize the support vertices. 
At each iteration, we also iterate the interlinear edges. Therefore, this phase takes $\OO(r^2m)$ time.
Same reasoning is also true for $k$ coplanar groups.
Since these operations are not nested, the overall asymptotic running time of Algorithm~\ref{alg:hyperplanarlocalization} is $\OO(k^2m) + \OO(r^2m) + O(n+m)$.

\section{Conclusion}
This paper describes an energy-efficient algorithm to localize a wireless sensor network.
In general, localization problem is an \NP-hard problem. 
However, we have utilized the observation that in some cases the sensor nodes form hyperplanar groups.
Using this information, we have proposed methods to localize a network, whose localization will consume much more time using
conventional trilateration and quadrilateration algorithms.

We have defined a function $h$, called hyperplanar grouping function, which maps the vertices to hyperplanar groups.
Our algorithms work under the assumption of availability of the hyperplanar grouping function, since we proved that determining that function is an \NP-hard problem.

Instead of iterating each vertex multiple times, our proposed algorithm breaks down the localization problem into lower-dimensional (and also smaller) problems, and thus reducing the quartic running time complexity with respect to the number of vertices.
Since the number of hyperplanar groups is assumed to be much less than the number of vertices, the quadratic running time on the number of groups amortize itself as more crowded the groups get.

In addition, to the best of our knowledge, this paper is the first in the literature that proves unit interval graphs are (claw,net)-free graphs.
Moreover, using this information, we have described a linear-time algorithm to embed a given unit interval graph on a straight line.

\balance
\bibliographystyle{abbrv}
\bibliography{ref}  

\begin{thebibliography}{10}

\bibitem{dwrl}
H.~Akcan and C.~Evrendilek.
\newblock {GPS-free directional localization via dual wireless radios}.
\newblock {\em Computer Communications}, 35(9):1151--1163, 2012.

\bibitem{reducing}
H.~Akcan and C.~Evrendilek.
\newblock {Reducing The Number Of Flips In Trilateration With Noisy Range
  Measurements}.
\newblock In {\em Proceedings of the 12th ACM International Workshop on Data
  Engineering for Wireless and Mobile Access (MobiDE'13)}, pages 20--27, New
  York, NY, USA, June 2013.

\bibitem{survey}
I.~F. Akyildiz, W.~Su, Y.~Sankarasubramaniam, and E.~Cayirci.
\newblock {Wireless Sensor Networks: A Survey}.
\newblock {\em Comput. Netw.}, 38(4):393--422, Mar. 2002.

\bibitem{theory}
J.~Aspnes, T.~Eren, D.~Goldenberg, A.~Morse, W.~Whiteley, Y.~Yang, B.~Anderson,
  and P.~Belhumeur.
\newblock {A theory of network localization}.
\newblock In {\em {IEEE} Transactions on Mobile Computing}, volume~5, pages
  1663--1678, 2010.

\bibitem{complexityOfWSN}
J.~Aspnes, D.~Goldenberg, and Y.~R. Yang.
\newblock {\em On the Computational Complexity of Sensor Network Localization},
  pages 32--44.
\newblock Springer Berlin Heidelberg, Berlin, Heidelberg, 2004.

\bibitem{linearhamiltonian}
A.~Brandst{\"a}dt, F.~F. Dragan, and E.~K{\"o}hler.
\newblock {\em {Linear Time Algorithms for Hamiltonian Problems on
  (Claw,Net)---Free Graphs}}, pages 364--376.
\newblock Springer Berlin Heidelberg, 1999.

\bibitem{udgRecognition}
H.~Breu and D.~G. Kirkpatrick.
\newblock Unit disk graph recognition is {NP}-hard.
\newblock {\em Computational Geometry}, 9(1):3 -- 24, 1998.
\newblock Special Issue on Geometric Representations of Graphs.

\bibitem{gpslesslowcost}
N.~Bulusu, J.~Heidemann, and D.~Estrin.
\newblock {GPS-less low-cost outdoor localization for very small devices}.
\newblock {\em Personal Communications, IEEE}, 7(5):28--34, Oct 2000.

\bibitem{cagiricithesis}
O.~Cagirici.
\newblock {Exploiting Coplanar Clusters to Enhance 3D Localization in Wireless
  Sensor Networks}.
\newblock Master's thesis, The Graduate School of Natural and Applied Sciences
  of Izmir University of Economics, Jan 2015.

\bibitem{udg}
B.~Clark, C.~Colbourn, and D.~Johnson.
\newblock {Unit Disk Graphs}.
\newblock {\em Discrete Mathematics}, 86(13):165--177, 1990.

\bibitem{localizabilityanddistributed}
Y.~Diao, Z.~Lin, M.~Fu, and H.~Zhang.
\newblock {Localizability and Distributed Localization of Sensor Networks using
  Relative Position Measurements}.
\newblock {\em IFAC Proceedings Volumes}, 46(13):1--6, 2013.

\bibitem{rangebased}
B.~Dil, S.~Dulman, and P.~Havinga.
\newblock {\em Range-Based Localization in Mobile Sensor Networks}, volume 3868
  of {\em Lecture Notes in Computer Science}.
\newblock Springer Berlin Heidelberg, 2006.

\bibitem{trilat}
H.~L. Groginsky.
\newblock {Position Estimation Using Only Multiple Simultaneous Range
  Measurements}.
\newblock {\em Aeronautical and Navigational Electronics, IRE Transactions on},
  6(3):178--187, Sept 1959.

\bibitem{gps}
B.~Hofmann-Wellenhof, H.~Lichtenegger, and J.~Collins.
\newblock {\em Range-Based Localization in Mobile Sensor Networks}, volume 3868
  of {\em Lecture Notes in Computer Science}.
\newblock Springer Wien, 1997.

\bibitem{unitdiskapprox}
F.~Kuhn, T.~Moscibroda, and R.~Wattenhofer.
\newblock {Unit Disk Graph Approximation}.
\newblock In {\em Proceedings of the 2004 Joint Workshop on Foundations of
  Mobile Computing}, DIALM-POMC '04, pages 17--23, 2004.

\bibitem{laman}
G.~Laman.
\newblock {On graphs and rigidity of plane skeletal structures}.
\newblock {\em Journal of Engineering Mathematics}, 4(10):331 -- 340, 2002.

\bibitem{techniques}
G.~Mao, B.~Fidan, and B.~Anderson.
\newblock {Wireless sensor network localization techniques}.
\newblock {\em Computer Networks}, 51(10):2529--2553, 2007.

\bibitem{surveypal}
A.~Pal.
\newblock {Localization Algorithms in Wireless Sensor Networks: Current
  Approaches and Future Challenges}.
\newblock {\em Network Protocols and Algorithms}, 2(1):45--73, 2010.

\bibitem{understandinglocalizability}
Z.~Yang and Y.~Liu.
\newblock Understanding node localizability of wireless ad-hoc networks.
\newblock In {\em 2010 Proceedings IEEE INFOCOM}, pages 1--9, March 2010.

\bibitem{zhongthesis}
Z.~Zhong.
\newblock {\em {Range-Free Localization and Tracking in Wireless Sensor
  Networks}}.
\newblock PhD thesis, The Graduate School of the University of Minnesota, Sep
  2010.

\end{thebibliography}

\end{document}